\documentclass[a4paper,12pt]{article}

\usepackage{amsthm,amsfonts,amsmath,enumerate,a4wide}

\title{On the Morse-Hedlund complexity gap}
\author{Julien Cassaigne \and Fran\c{c}ois Nicolas}

\newcommand{\N}{\mathbb{N}} 
 \newcommand{\Nast}{\N \setminus \{ 0 \}} 
\newcommand{\lgr}[1]{\left| #1 \right|}

\newcommand{\ze}{\mathtt{a}}
\newcommand{\on}{\mathtt{b}}
\newcommand{\mv}{\varepsilon}

\newcommand{\zeon}{\left\{ \ze, \on \right\}}

\newtheorem{lemma}{Lemma}
\newtheorem{theorem}{Theorem}
\newtheorem{definition}{Definition}
\newtheorem{claim}{Claim}

\theoremstyle{definition}
\newtheorem{exercise}{Exercise}
\newtheorem{example}{Example}

\newcommand{\fiatc}{\textup{FIATC}}

\begin{document}

\maketitle

\sloppy

\begin{abstract}
In 1938, Morse and Hedlund proved that the subword complexity function of any bi-infinite word is either bounded or at least linearly growing.
In 1982, Ehrenfeucht and Rozenberg proved that this gap property holds for the subword complexity function of any language.
The aim of the present paper is to present a self-contained, compact proof of the latter  result.
\end{abstract}

\section{Notation and definitions}

The set of natural integers is denoted $\N$.
Throughout the paper, 
$A$ denotes a finite set of symbols, \emph{i.e.}, an \emph{alphabet},
and 
$\alpha$ denotes the cardinality of $A$.

A \emph{word} over $A$ is a finite string of elements of $A$. 
The set of all words over $A$ is denoted $A^\star$.
For every $w \in A^\star$,  $\lgr{w}$ denotes the \emph{length} of $w$. 
For each $n \in \N$,  $A^n$ denotes the set of all $n$-length words over $A$.
The \emph{empty word}, denoted $\mv$, is the unique word of length zero.
Word concatenation is denoted multiplicatively.
Given $x$, $y \in A^\star$,
we say that $x$ is a \emph{prefix} of $y$ if there exists $w \in A^\star$ such that $y = xw$,
we say that $x$ is a \emph{suffix} of $y$ if there exists $w \in A^\star$ such that $y = wx$,  and
we say that $x$ is a \emph{factor} of $y$ if there exist $w$, $w' \in A^\star$ such that $y =  w x w'$.

A \emph{language} over $A$ is a subset of $A^\star$.
Let $L$ be a language over $A$.
Given $x \in A^\star$, 
we say that $x$ is a factor of $L$ if $x$ is a factor of some word in $L$.
For each $n \in \N$, $F_n(L)$ denotes the set of all $n$-length factors of $L$.
The \emph{complexity function} of $L$ is the function mapping each $n \in \N$ to the cardinality of $F_n(L)$.
Let $p$ denote the complexity function of $L$.
Note that $p(0) = 1$ unless $L = \emptyset$ because $\mv$ is a factor of every word.

\begin{example}
Let $L = A^\star$.
For each $n \in \N$,
it is clear that $F_n(L) = A^n$, 
so $p(n) = \alpha^n$. 
\end{example}

\begin{example} \label{ex:aabb}
Let $L = \left\{ \ze^i \on^j : i, j \in \N  \right\}$.
For each $n \in \N$, 
it is clear that $F_n(L) = \left\{ \ze^{n - k} \on^k : k = 0, 1, 2, \dotsc, n \right\}$, 
so $p(n) = n + 1$.
\end{example}

\begin{example}
Let $L = \left\{ \on^i \ze \on^j \ze \on^i : i, j \in \N  \right\}$.
For each $n \in \N$, 
$F_n(L)$ is the set of those $n$-length words over $\{ \ze, \on \}$ in which $\ze$ occurs at most twice, 
so $p(n) = \frac{1}{2}n(n +  1) + 1$.  
\end{example}

\begin{example} \label{ex:bounded-non-constant}
Let 
$L = \left\{ \on \ze^{2k} \on : k \in \N  \right\}$ 
and 
$X_n = \left\{ \ze^n, \on \ze^{n - 1}, \ze^{n - 1} \on \right\}$
for each $n \in \Nast$.
For each $k \in \Nast$, 
it is clear that 
$F_{2k}(L) = X_{2k} \cup \left\{ \on \ze^{2k - 2} \on  \right\}$ 
and 
$F_{2k + 1}(L) = X_{2k + 1}$,
so $p(2k) = 4$ and $p(2k + 1) = 3$.
\end{example}

\begin{exercise} 
Let $L = \left\{ aa : a \in A \right\}^\star$: $L$ is the closure of $\left\{ aa : a \in A \right\} \cup \{ \mv \}$ under concatenation.
Prove that 
$p(2k + 1) = 2 \alpha^{k + 1} - \alpha$ 
and 
$p(2k) = \alpha^{k + 1} + \alpha^k  - \alpha$ for every $k \in \N$.
\end{exercise}


\begin{exercise} \label{exo:exp} 
Prove that if $L$ is a proper subset of $A^\star$ then there exist real numbers $\lambda$, $\beta \ge 1$ such that 
$\beta < \alpha$ 
and 
$p(n) \le \lambda \beta^n$ for every $n \in \N$ (hint: $p(n + n') \le p(n) p(n')$ for all $n$, $n' \in \N$). 
\end{exercise}

\section{The Morse-Hedlund complexity gap} 

The aim of this paper is to present a self-contained, compact proof of:

\begin{theorem}[Ehrenfeucht and Rozenberg, 1982 \cite{EhrenfeuchtR82}] \label{th:ER82}
Let  $p$ be the complexity function of some language. 
Either $p(n)$ is greater than $n$ for every $n \in \N$, or $p$ is bounded.
\end{theorem}

For instance, it follows from Theorem~\ref{th:ER82} that no complexity function grows like $\sqrt{n}$.
Example~\ref{ex:aabb} shows that the lower bound is tight.
Exercise~\ref{exo:exp} shows that there is also a ``complexity gap'' at the opposite extremity of the spectrum.
In addition to proving Theorem~\ref{th:ER82}, 
Ehrenfeucht and Rozenberg described the class of those languages whose complexity functions are bounded:

\begin{theorem}[Ehrenfeucht and Rozenberg, 1982 \cite{EhrenfeuchtR82}] \label{th:ER82-v}
Let $L$ be a language over $A$.
The complexity function of $L$ is bounded if, and only if, 
 there exists a finite language $X$ over $A$ such that 
$L \subseteq   \left\{ x y^n z :  (x, y, z, n) \in X \times X \times X \times \N \right\}$.
\end{theorem}

Before presenting the proof of Theorem~\ref{th:ER82}, 
let us state the various related results that can be found in the literature.
In order to do so, 
let us introduce a little more material.
A \emph{(right-)infinite word} over $A$  is a function from $\N$ to $A$.
A \emph{bi-infinite} word over $A$ is a function from the set of rational integers to $A$.
Let  $u$ be an infinite or bi-infinite word over $A$.
Given $x \in A^\star$,
we say that $x$ is a factor of $u$ if there exists $i$ in the domain of $u$ such that 
$x = u(i) u(i + 1) u(i + 2) \dotsb u(i + \lgr{x} - 1)$.
The set of all factors of $u$ is called the language of $u$.
The complexity function of $u$ is defined as the complexity function of its language.

\begin{example} \label{ex:u-aabb}
Let $u$ be the bi-infinite word over $\{ \ze, \on \}$ given by: 
$u(i) = \on$ for every $i \in \N$ 
and 
$u(-i) = \ze$ for every $i \in \Nast$.
Let $L$ be as in Example~\ref{ex:aabb}.
The language of $u$ equals $L$, so the complexity function of $u$ maps $n$ to $n + 1$ for every $n \in \N$.
\end{example}

We say that a function $p\colon\N \to \N$ is \emph{\fiatc} (first increasing and then constant) if there exists $m \in \N$ such that 
$p(0) < p(1) < p(2) < \dotsb < p(m)$ and
$p(m + n) = p(m)$ for every $n \in \N$. 

\begin{theorem}[Morse and Hedlund, 1938 \cite{MorseHedlund38,CovenH1973}] \label{th:MH1938}
Let $u$ be a bi-infinite word and let $p$ denote the complexity function of $u$.
\begin{itemize}
\item 
If $u$ is not periodic then $p$ is increasing.
\item 
If $u$ is periodic then $p$ is \fiatc{} and $\sup_{n \in \N} p(n)$ is the least period of $u$.
\end{itemize}
\end{theorem}

On the one hand,  \fiatc{} functions are clearly bounded.
On the other hand, observe that any increasing function $p\colon \N \to \N$ satisfies $p(n) \ge n + p(0)$ for every $n \in \N$. 
Hence, Theorem~\ref{th:MH1938} implies Theorem~\ref{th:ER82} in the case where $p$ is the complexity function of a bi-infinite word.
Example~\ref{ex:bounded-non-constant} shows that bounded complexity functions are not necessarily \fiatc;
unbounded complexity functions are not necessarily monotonically increasing:

\begin{example} \label{ex:unbounded-non-inc} 
Let $L = \zeon^\star \cup \left\{ \mathtt{c}, \mathtt{d}, \mathtt{e} \right\}$. 
It is clear that $p(n) = 2^n$ for every $n \in \N \setminus \{ 1 \}$ 
and 
that $p(1) = 5$. 
Remark that $p$ is unbounded and $p(2) = 4 < 5 = p(1)$.
\end{example}

\begin{example}
Let 
$L = \left\{ \ze^k \on^k : k \in \N  \right\} 
\cup  \left\{ \on \ze^{2k} \on : k \in \N  \right\} 
\cup  \left\{ \ze \on^{2k} \ze : k \in \N  \right\}$. 
It is clear that $p(2k + 1) = 2k + 2$ for every $k \in \N$
and that  $p(2k) = 2k + 3$ for every $k \in \N \setminus \{ 0, 1 \}$.
Remark that $p$ is unbounded and $p(n) > p(n + 1)$ for infinitely many $n \in \N$.
\end{example}

\begin{exercise}
 Let $L$ be the set of all non-empty words $w$ over $\zeon$ such that for every integer $k$ with 
$1 \le k < 2^{\left\lfloor \log_4 \lgr{w} \right\rfloor}$,
$\on \ze^k \on$ is not a factor of $w$.
\begin{enumerate}
\item 
Prove that  $p$ is not polynomially bounded (hint: $p(4^k) \ge 2^{2^k}$ for every $k \in \N$).
\item 
Prove that $p(n) > p(n + 1)$ for infinitely many $n \in \N$ (hint: $p(4^k - 1) > p(4^k)$ for all but finitely many $k \in \N$).
\end{enumerate}
\end{exercise}

The most famous variant of Theorems~\ref{th:ER82}/\ref{th:ER82-v} and~\ref{th:MH1938} is: 

\begin{theorem}[\cite{CovenH1973,Bush1955,Lothaire2,Pytheas,CantBook}] \label{th:MH1973}
Let $u$ be an infinite word and let $p$ denote the complexity function of $u$.
\begin{itemize}
 \item 
If $u$ is not eventually periodic then $p$ is increasing.
 \item
If $u$ is eventually periodic then $p$ is \fiatc{} and the period of $u$ is not greater than $\sup_{n \in \N} p(n)$.
\end{itemize} 
\end{theorem}

A \emph{Sturmian word} is nowadays usually defined as an infinite word whose complexity function maps $n$ to $n + 1$ for every $n \in \N$ \cite{Lothaire2,Pytheas}.
By Theorem~\ref{th:MH1973}, Sturmian words are those non-eventually-periodic infinite words with minimum complexity.
There is no trivial example of Sturmian word.

\begin{exercise}
 Let $u$ be a Sturmian word and let $u'$ be the infinite word defined by: 
$u'(i) = u(i + 1)$ for every $i \in \N$.
Prove that the language of $u$ equals the language of $u'$.
\end{exercise}

To conclude the section, let us state the latest improvement of Theorem~\ref{th:ER82}.

\newcommand{\qq}[1]{\left\lceil \dfrac{#1 + 1}{2} \right\rceil \left\lfloor \dfrac{#1 + 1}{2} \right\rfloor}

\begin{theorem}[Balogh and Bollob\'as, 2005 \cite{BaloghB05}] \label{th:BB2005}
For each real number $x$, let 
$$
\phi(x) = \qq{x} \, .
$$
\begin{itemize}
\item 
Let $p$ be the complexity function of some language and let $m \in \N$.
If $p(m) \le m$ then
$p(n + p(m) + m) \le \phi(p(m))$ 
for every $n \in \N$.
\item  
For each $k \in \N$, there exists a function $p_k\colon \N \to \N$ such that
$p_k$ is the complexity function of some language  over $\{ \ze, \on \}$
and 
both sets $\left\{ n \in \N : p_k(n) = k \right\}$ and  $\left\{ n \in \N : p_k(n) = \phi(k) \right\}$ are infinite.
\end{itemize}
\end{theorem}

The second part of Theorem~\ref{th:BB2005} ensures that the function $\phi$ is optimal.

\section{The proof of Theorem~\ref{th:ER82}} \label{sec:proof}


\begin{definition}
We say that a word $w \in A^\star$ is a \emph{(right-)special factor} of the language $L$ if there exist $a$, $b \in A$ with $a \ne b$ such that both $wa$ and $wb$ are factors of $L$. 
\end{definition}

Special factors are standard tools for studying complexity functions \cite{CantBook}.

\begin{lemma} \label{lem:infinite-spec}
The language $L$ admits infinitely many special factors 
if, and only if,
for each $n \in \N$, $L$ admits at least one $n$-length special factor.
\end{lemma}

\begin{proof}
Simply remark that any suffix of any special factor of $L$ is also a special factor of~$L$.
\end{proof}

\begin{definition}
We say that the language $L$ is \emph{(right-)extendable} if for each $w \in L$, there exists $a \in A$ such that $wa \in L$.
\end{definition}

\begin{example} \label{ex:inf-ext} 
 The language of any infinite or bi-infinite word is extendable.
\end{example}

 \begin{lemma} \label{lem:const-inc} 
\leavevmode
\begin{enumerate} 
 \item \label{lem:finite-constant} If a language only admits finitely many special factors then its complexity function is eventually constant.
 \item If an extendable language admits infinitely many special factors then its complexity function is increasing.
\end{enumerate}
\end{lemma}

\begin{proof}
For each non-empty word $w$, let $\rho(w)$ denote the ${(\lgr{w} - 1)}$-length prefix of $w$:
$\rho(xa) = x$ for every $x \in A^\star$ and every $a \in A$.
For every $w \in A^\star$, 
$w$ is a special factor of $L$ if, and only if, there exist $u$, $v \in L$ such that $u \ne v$ and $\rho(u) = \rho(v) = w$.

Assume that $L$ only admits finitely many special factors. 
Let $m$ be an upper bound on the lengths of the special factors of $L$.
For every integer $n > m$,
$\rho$ induces an injection from $F_{n + 1}(L)$ into $F_n(L)$, and thus we have $p(n + 1) \le p(n)$.
Therefore, $p$ is monotonically decreasing on $\{ n \in \N : n > m  \}$. 
It follows that $p$ is eventually constant.

Assume that $L$ is extendable and admits  infinitely many special factors.
Let $n \in \N$.
On the one hand,  
$\rho$ induces a surjection from $F_{n + 1}(L)$ onto $F_n(L)$ because $L$ is extendable.
On the other hand, 
$\rho$ is not injective on $F_{n + 1}(L)$ because, 
by Lemma~\ref{lem:infinite-spec},
$L$ admits at least one $n$-length special factor. 
Hence, $\rho$ induces a non-bijective surjection from $F_{n + 1}(L)$ onto $F_n(L)$.
It follows $p(n + 1) > p(n)$.
We have thus shown that $p$ is increasing.
 \end{proof}

Note that the converse of Lemma~\ref{lem:const-inc}.\ref{lem:finite-constant}  does not hold in general:

\begin{example}
Let $L = \{ \ze^k \on : k \in \N \}$.
The complexity function of $L$ is eventually constant: 
for every $n \in \Nast$, 
we have $p(n) = 2$ because $F_n(L) =  \{ \ze^n,  \ze^{n - 1} \on \}$.
However, $L$ admits infinitely many special factors: 
for every $n \in \N$, $\ze^n$ is a special factor of $L$. 
\end{example}

\begin{exercise} 
Prove that the complexity function of any extendable language is either increasing or \fiatc. 
\end{exercise}

\begin{exercise} \label{exo:eventually-w}
 Prove that if the language of an infinite word only admits finitely many special factors then this infinite word is eventually periodic.
\end{exercise}

\begin{exercise} \label{exo:eventually-lg}
Prove that if the language $L$ only admits finitely many special factors then there exists a finite language $X$ over $A$ such that 
$
L \subseteq \left\{ x y^n : (x, y, n) \in X \times X \times \N \right\} 
$.
\end{exercise}

\begin{exercise} \label{exo:p-s}
For each $n \in \N$, let $s(n)$ denote the number of $n$-length special factors of $L$.
\begin{enumerate}
 \item  
Prove that $p(n + 1) - p(n) \le (\alpha - 1) s(n)$ for every $n \in \N$.
\item Prove that if $L$ is extendable then $p(n + 1) - p(n) \ge s(n)$ for every $n \in \N$.
\end{enumerate}
\end{exercise}

Lemma~\ref{lem:const-inc} can be easily deduced from Exercise~\ref{exo:p-s}.
%

\begin{definition}
A language is called \emph{factorial} if it equals its set of factors.
\end{definition}


\begin{proof}[Proof of Theorem~\ref{th:ER82}]
For every $k \in \N$ and every $w \in A^\star$, 
let $E_k(w)$ denote the set of all $x \in A^k$ such that $wx$ is a factor of $L$.
For every $k \in \N$, 
let $L_k$ denote the set of all $w \in A^\star$ such that $E_k(w) \ne \emptyset$.
Put $L_\infty = \bigcap_{k \in \N} L_k$.

\begin{claim} \label{claim:Lk-facto}
 For every $k \in \N \cup \{ \infty \}$, the language $L_k$ is factorial.
\end{claim}

\begin{proof}
Let $k \in \N$.
Let $w$ be a factor of $L_k$.
There exist $v$, $v' \in A^\star$  such that $v w v' \in L_k$.
For any $x \in E_k(v w v')$, the $k$-length prefix of $v' x$ belongs to $E_k(w)$.
Therefore, $w$ belongs to $L_k$.
We have thus shown that $L_k$ is factorial for every $k \in \N$.

Moreover, any intersection of factorial languages is also a factorial language.
Therefore, $L_\infty$ is factorial.
\end{proof}

\begin{claim}  \label{claim:Lk+1-Lk}
For every $k \in \N$, $L_{k + 1}$ is a subset of $L_k$.
\end{claim}

\begin{proof}
Let $w \in L_{k + 1}$.
For any $x \in E_{k + 1}(w)$, 
the $k$-length prefix of $x$ belongs to $E_k(w)$.
Therefore, $w$ belongs to $L_k$.
\end{proof}

\begin{claim}  \label{claim:Lk-Lk+1A}
For every $k \in \N$, $L_k \setminus \{ \mv \}$ is a subset of $L_{k + 1} A$.
\end{claim}

\begin{proof}
Let $w \in L_k \setminus \{ \mv \}$.
Write $w$ in the form $w = w' a$ with $w' \in A^\star$ and $a \in A$.
For any $x \in E_k(w)$, $ax$ belongs to  $E_{k + 1}(w')$.
It follows that $w'$ belongs to $L_{k + 1}$,
and thus $w$ belongs to $L_{k + 1} A$.
\end{proof}

\begin{claim} \label{claim:Linfty-ext}
 The language $L_\infty$ is extendable.
\end{claim}

 \begin{proof}
Let $w \in L_\infty$.
For each $k \in \N$, there exists $a_k \in A$ such that $wa_k \in L_k$: 
the first letter of any element of $E_{k + 1}(w)$ is a suitable choice for $a_k$.
It follows that there exists $a \in A$ such that $wa \in L_k$ for infinitely many $k \in \N$.
Since the sequence of sets $\left(L_k \right)_{k \in \N}$ is monotonically decreasing by Claim~\ref{claim:Lk+1-Lk}, $wa$ belongs to $L_\infty$.
\end{proof}

For each $k \in \N \cup \{ \infty \}$, let $p_k$ denote the complexity function of $L_k$.
To prove Theorem~\ref{th:ER82}, it suffices to check that the following four assertions are equivalent:
\begin{enumerate}[$(i)$.]
\item \label{item:pm-m} 
There exists $m \in \N$ such that $p(m) \le m$.
\item \label{item:pinfty-bounded}
 $L_\infty$ only admits finitely many special factors. 
\item \label{item:pj-bounded} 
There exists $j \in \N$ such that  $L_j$ only admits finitely many special factors.
\item \label{item:p-bounded}
$p$ is bounded.
\end{enumerate}

\paragraph{$(\ref{item:p-bounded}) \implies (\ref{item:pm-m})$.} Assume that $p$ is bounded.
Put $m =  \sup_{n \in \N} p(n)$.
Clearly, $m$ satisfies $p(m) \le m$.
Therefore, implication $(\ref{item:p-bounded}) \implies (\ref{item:pm-m})$ holds.

\paragraph{$(\ref{item:pj-bounded}) \implies (\ref{item:p-bounded})$.}
For all  $k$, $n \in \N$,
$F_{n + 1}(L_{k})$ is a subset of $F_n(L_{k + 1}) A$ by Claim~\ref{claim:Lk-Lk+1A},
and thus inequality $p_k(n + 1) \le \alpha p_{k + 1}(n)$ holds.
Hence, for every $k \in \N$, $p_{k + 1}$ is bounded only if $p_{k}$ is bounded.
Now, assume that there exists $j \in \N$ such that  $L_j$ only admits finitely many special factors.
Then,  $p_j$ is bounded by Lemma~\ref{lem:const-inc}.\ref{lem:finite-constant}.
It follows that the $j$ complexity functions $p_{j - 1}$, $p_{j - 2}$, $p_{j - 3}$, \ldots, $p_0$ are also bounded.
Besides, $p$ equals $p_0$ because $L_0$ equals the set of all factors of $L$.
We have thus proven that $p$ is bounded.
The proof of implication $(\ref{item:pj-bounded}) \implies (\ref{item:p-bounded})$ is complete.

\paragraph{$(\ref{item:pinfty-bounded}) \implies (\ref{item:pj-bounded})$.}
Assume that $L_\infty$ only admits finitely many special factors.
Let $n \in \N$ be such that $L_\infty$ does not admit any $n$-length special factor.
On the one hand,
it follows from Claim~\ref{claim:Lk-facto} that $F_{n + 1}(L_k) = L_k \cap A^{n + 1}$ for
every $k \in \N \cup \{ \infty \}$, and thus
$$
F_{n + 1}(L_\infty) 
= 
L_\infty \cap A^{n + 1}
= 
\bigcap_{k \in \N} (L_k \cap  A^{n + 1} ) 
=  
\bigcap_{k \in \N} F_{n + 1}(L_k) \, . 
$$
On the other hand,
$\left(F_{n + 1}(L_k) \right)_{k \in \N}$ 
is a sequence of finite sets, which is monotonically decreasing by Claim~\ref{claim:Lk+1-Lk}, 
so there exists $j \in \N$ such that
$$
F_{n + 1}(L_j) = \bigcap_{k \in \N} F_{n + 1}(L_k) \, .
$$
Hence, we have $F_{n + 1}(L_j) = F_{n + 1}(L_\infty)$, 
and thus $L_j$ does not admit any $n$-length special factor.
It then follows from Lemma~\ref{lem:infinite-spec}  that $L_j$ only admits finitely many special factors.
The proof of implication $(\ref{item:pinfty-bounded}) \implies (\ref{item:pj-bounded})$ is complete.

\paragraph{$(\ref{item:pm-m}) \implies (\ref{item:pinfty-bounded})$.}
Assume that there exists $m \in \N$ such that $p(m) \le m$.
Since $L_\infty$ is a subset of $L_0$, we have $p_\infty(m) \le p_0(m) = p(m) \le m$,
so $p_\infty$ is not increasing.
It then follows from Claim~\ref{claim:Linfty-ext} and Lemma~\ref{lem:const-inc} 
that $L_\infty$ only admits finitely many special factors.
We have thus proven implication $(\ref{item:pm-m}) \implies (\ref{item:pinfty-bounded})$.
\end{proof}

\begin{exercise}
\label{exo:cpx-bounded}
Prove Theorem~\ref{th:ER82-v} (hint: use Exercise~\ref{exo:eventually-lg}).
\end{exercise}

\bibliographystyle{plain}
\bibliography{BibPansiot}

\end{document}